\documentclass{article}

\usepackage{hanrsp,amsmath,amssymb,theorem,amscd,upgreek,natbib,enumerate}
\newcommand{\tens}[1]{\mathsf{#1}}
\renewcommand{\vec}[1]{\mathbf{#1}}
\newcommand{\mat}[1]{\mathbf{#1}}

\newcommand\qedsymbol{\hbox{\rlap{$\sqcap$}$\sqcup$}}

\newtheorem{theorem}{Theorem}[section]
\newtheorem{definition}[theorem]{Definition}
\newtheorem{corollary}[theorem]{Corollary}
\newtheorem{lemma}[theorem]{Lemma}
\theoremstyle{plain}

\newtheorem{remark}[theorem]{Remark.}

\newenvironment{proof}{Proof.}{\hfill\qedsymbol\newline}

\newcommand{\proofthname}{Proof of the theorem}
\newenvironment{prooftheorem}{\proofthname}{\hfill\qedsymbol\newline}

\newcommand{\im}{\mathop{\mathrm{Im}}}
\newcommand{\re}{\mathop{\mathrm{Re}}}

\newcommand{\D}{\mathrm{D}}

\renewcommand{\Sigma}{\mathfrak{S}}

\newcommand{\I}{\mathrm{I}}

\newcommand{\upi}{\uppi}

\author{M. Seredy\'{n}ska\\
Institute of Fundamental Technological Research\\
Polish Academy of Sciences\\
 ul. Pawi\'{n}skiego 5b\\
02-106 Warszawa, PL\\
email: {\tt msered@ippt.gov.pl}\\ \vspace{0.2cm}\\
A. Hanyga\\
ul. Bitwy Warszawskiej 14 m. 52\\
02-366 Warszawa, PL\\
email: {\tt ajhbergen@yahoo.com}
}

\title{Positive solutions of viscoelastic problems} 

\begin{document}

\maketitle

\noindent\textbf{Keywords:} viscoelasticity, Newtonian viscosity, relaxation, 
completely monotone function, complete Bernstein function, MacDonald function 
\begin{abstract}
In 1,2 or 3 dimensions a scalar wave excited by a 
non-negative source
 in a viscoelastic medium with a non-negative relaxation spectrum or 
a Newtonian response or both combined inherits the sign of the source. The
key assumption is a constitutive relation which involves the sum of 
a Newtonian viscosity term and a memory term with a completely monotone relaxation kernel.
In higher-dimensional spaces this result holds for sufficiently regular sources. 
Two positivity results for vector-valued wave fields including isotropic viscoelasticity are also obtained.
\end{abstract}

\vspace{0.5cm}

\noindent\textbf{Notation.}

\noindent $[a,b\,[\; := \{ x \in \mathbb{R} \mid a \leq x < b \}$;\\
$\mathbb{R}_+ =\; ]0,\infty[\;$;\\
$\mat{I}$: unit matrix; \\
$\langle \vec{k}, x\rangle := \sum_{n=1}^d k_l\, x^l$\,;\\
$\theta(y) := \left\{ \begin{array}{ll} 1, & y > 0 \\
0,  & y < 0 \end{array} \right\} $;\\
$\tilde{f}(p) := \int_0^\infty \e^{- p y} \, f(y) \, \dd y$.

\section{Introduction}

Positivity of viscoelastic pulses was studied in a paper of \cite{Duff69}. 
Duff assumed a special model with a rational complex modulus. Duff's 
models are however loosely related to viscoelasticity and his assumptions are 
excessively restrictive.

In this
paper a general scalar viscoelastic medium with the constitutive equation 
$\sigma = a \, \dot{e} + G(t)\ast \dot{e}$ with a completely monotone
relaxation modulus $G$ and a non-negative Newtonian viscosity coefficient 
is studied. We show that a scalar viscoelastic wave field 
propagating in a $d$-dimensional medium and 
excited by a non-negative pulse is also non-negative provided $d \leq 3$.
For higher dimensions and for non-zero initial data only wave fields excited by 
sufficiently regular sources are non-negative.

Positivity of 
viscoelastic signals can be considered as a test for the 
non-negative relaxation spectrum and for the presence of the Newtonian 
viscosity.

Positivity can be extended to matrix-valued fields, e.g. to Green's functions
of systems of PDEs. In Sec.~\ref{sec:vector} we consider a system of PDEs resembling the equations 
of motion of viscoelasticity with a CM relaxation kernel and prove that the Green's
function of this system of equations is positive-semidefinite. This result 
does not apply to general viscoelastic Green's functions, which involve 
double gradients of positive semi-definite functions. In iso


\section{Statement of the problem}

In a hereditary or Newtonian linear viscoelastic medium a scalar field 
excited by positive source is non-negative. This applies to displacements 
in pure shear or to scalar displacement potentials. The key assumption about 
the material properties of the medium is a positive relaxation spectrum.
The result holds for arbitrary spatial dimension. 

We consider the problem:
\begin{equation} \label{eq:problem}
\rho\, \D^2 \, u = a \, \nabla^2 \, \D\, u + G(t)\ast \nabla^2\,\D u + 
s(t,x)
 \qquad t \geq 0, \quad x \in \mathbb{R}^d 
\end{equation}
with $s(t,x) = \theta(t)\,(c_1 + c_2 \, t) \,\delta(x)$ and the 
initial condition 
\begin{equation} \label{eq:ini}
u(0,x) = u_0\, \delta(x), \quad \D u(0,x) = \dot{u}_0 \,\delta(x),  
\end{equation}
(Problem~I) as well as $s(t,x) = c\,\delta(t) \, \delta(x)$ with a solution assumed to 
vanish for $t < 0$ (Problem~II).
It is assumed that $a \geq 0$ and $G$ is a completely monotone (CM) function.

The Laplace transform 
\begin{equation}
\tilde{u}(p,x) := \int_0^\infty \e^{-p t}\, u(t,x)\, \dd t, \qquad \re p > 0,
\quad x \in \mathbb{R}^d
\end{equation}
satisfies the equation
\begin{equation} \label{eq:Laplacetransformed}
\rho\, p^2 \, \tilde{u}(p,x) = Q(p) \, \nabla^2\, \tilde{u}(p,x) 
+ g(p) \, \delta(x)
\end{equation}
where
\begin{equation} \label{eq:defQ}
Q(p) := a \, p + p \, \tilde{G}(p)
\end{equation}
The function $g$ is defined by the equation
\begin{equation} 
g_I(p) = \frac{1}{p} + p\, u_0 + \dot{u}_0
\end{equation}
in Problem~I and
\begin{equation} 
g_{II}(p) = 1
\end{equation}
in Problem~II.

\section{Basic mathematical tools.}
\label{sec:tools}

The classes of functions appropriate for viscoelastic responses are 
reviewed in detail in \citet{HanSerLANL}.

\begin{theorem}
If the function $\tilde{u}(\cdot,x)$ is completely monotone for every 
$x \in \mathbb{R}^d$, then $u(t,x) \geq 0$ for every  $t \geq 0$ and
$x \in \mathbb{R}^d$.
\end{theorem}
\begin{proof}
If $\tilde{u}(\cdot,x)$ is completely monotone, then, in view of Bernstein's
theorem \citep{WidderLT}, for every $x \in \mathbb{R}^d$ it is the Laplace
transform of a positive Radon measure $m_x$:
\begin{equation} 
\tilde{u}(\cdot,x) = \int_{[0,\infty[} \e^{-p s}\, m_x(\dd s)
\end{equation}
The Radon measure $m_x$ is uniquely determined by $\tilde{u}(\cdot,x)$,
hence $m_x(\dd t) = u(t,x)\, \dd t$ is a positive Radon measure. Hence,
in view of continuity of $u(\cdot,x)$, we have the inequality 
$u(t,x) \geq 0$  for all $t \geq 0$ and $x \in \mathbb{R}^d$.  
\end{proof}

The problem of proving that $u(t,x)$ is non-negative is thus reduced to 
proving that $\tilde{u}(\cdot,x)$ is completely monotone. The crucial 
step here is the realization that $Q$ in \eqref{eq:defQ}
is a complete Bernstein function.
We shall therefore recall some facts about Bernstein and complete 
Bernstein functions and their relations to 
completely monotone functions.
\begin{definition}
A function $f$ on $\mathbb{R}_+$ is said to be \emph{completely monotone} (CM) if
it is infinitely differentiable and satisfies the infinite set of inequalities:
$$ (-1)^n \, \D^n \, f(y) \geq 0  \qquad y > 0, \quad 
\text{for all non-negative integer $n$}$$
\end{definition}
It follows from the definition and the Leibniz formula that the
product of two CM functions is CM.
A CM function can have a singularity at 0. 
\begin{definition}
A function $f$ on $\mathbb{R}_+$ is said to be \emph{locally integrable 
completely monotone} (LICM) if it is CM and integrable over the segment $]0,1]$.
\end{definition}
\begin{definition}
A function $f$ on $\mathbb{R}_+$ is said to be a \emph{Bernstein function} (BF) if
it is non-negative, differentiable and its derivative is a CM function.
\end{definition}
Since a BF is non-negative and non-decreasing, it has a finite limit at 0.
It can therefore be extended to a function on $\overline{\mathbb{R}_+}$.  

Every CM function $f$ is the Laplace transform of a positive Radon measure:
\begin{theorem}(Bernstein's theorem,\citet{WidderLT}) \label{thm:Bernstein}
\begin{equation}  \label{eq:Bernstein}
f(t) = \int_{[0,\infty[}\; \e^{-r t} \, \mu(\dd r)
\end{equation} 
\end{theorem}
It is easy to show that $f$ is a LICM if the Radon measure
$\mu$ satisfies the inequality
\begin{equation} \label{eq:ineq}
\int_{[0,\infty[}\; \frac{\mu(\dd r)}{1 + r} < \infty
\end{equation}

\begin{theorem} \label{thm:CMmult}\citep{Jacob01I,HanSerLANL}
If $f, g$ are CM then the pointwise product $f\, g$ is CM.
\end{theorem}

Let $f$ be a Bernstein function. Since the derivative $\D f$ of $f$
is LICM, Bernstein's theorem can be applied. Upon integration the following
integral representation of of a general Bernstein function $f$ is obtained:
\begin{equation}
f(y) = a + b\,y + \int_{]0,\infty[} \left[ 1 - \e^{-r \, y}\right] \,
\nu(\dd r) 
\end{equation}
where $a, b = \D f(0) \geq 0$, and $\nu(\dd r) := \mu(\dd r)/r$ is a 
positive Radon measure on $\mathbb{R}_+$ satisfying the inequality
\begin{equation} \label{eq:ineq2}
\int_{]0,\infty[}\; \frac{r \,\nu(\dd r)}{1 + r} < \infty
\end{equation}
The constants $a, b$ and the Radon measure $\nu$ are uniquely determined by the
function $f$.

\begin{theorem}\citep{BergForst,Jacob01I} \label{thm:compBFCM}
If $f$ is a CM function, $g$ is a BF and $g(y) > 0$ for
$y > 0$ then the composition $f\circ g$ is a CM.
\end{theorem} 

\begin{corollary}\citep{BergForst,Jacob01I} \label{corr:BF}
If $g$ is a non-zero BF then $1/g$ is a CM function.
\end{corollary} 

Note that the function $f(y) := \exp(-y)$ is CM but $1/f$ is not a BF.

\begin{definition} \label{def:CBF}
A function $f$ is said to be a \emph{complete Bernstein function} (CBF) if 
there is a Bernstein function $g$ such that $f(y) = y^2\, \tilde{g}(y)$.
\end{definition}
\begin{theorem} \label{thm:analCBF}  \citep{Jacob01I}
A function $f$ is a CBF if and only if it satisfies the following two 
conditions:
\begin{enumerate}
\item $f$ admits an analytic continuation $f(z)$ to the upper complex half-plane;
$f(z)$ is holomorphic and satisfies the inequality 
$\im f(z) \geq 0$ for $\im z > 0$;
\item $f(y) \geq 0$  for $y \in \mathbb{R}_+$.
\end{enumerate}
\end{theorem}

The derivative $\D g$ of the Bernstein function $g$ is a LICM function $h$. 
Hence we have the following theorem:
\begin{theorem}
Every CBF $f$ can be expressed in the form 
\begin{equation} \label{eq:CBFLICM}
f(y) = y \, \tilde{h}(y) + a\, y
\end{equation}
where $h$ is LICM and $a = g(0) \geq 0$. Conversely, for every LICM function $h$
and $a \geq 0$ the function $f$ given by \eqref{eq:CBFLICM} is a CBF. 
\end{theorem}
\begin{proof} 
For the first part, let $g$ be the BF in Definition~\ref{def:CBF} and let $h := \D g$. 
Since $\int_0^1 h(x) \, \dd x = g(1) - g(0) < \infty$, the function $h$ is LICM.
For the second part, note that if $h$ is LICM, then 
$g(y) = a + \int_0^y h(s) \, \dd s$ is a BF and $f(y) = y^2 \, \tilde{g}(y)$.
\end{proof}

Since the Laplace transform of a LICM function $h$ has the form 
\begin{equation}
\tilde{h}(y) = \int_{[0,\infty[}\; \frac{\mu(\dd r)}{r + y}
\end{equation}
where $\mu$ is the Radon measure associated with $h$,
every CBF function $f$ has the following integral representation 
\begin{equation} \label{eq:CBFintegral}
f(y) = b + a \, y + y \int_{]0,\infty[}\; \frac{\mu(\dd r)}{r + y}
\end{equation}
with arbitrary $a, b  = \mu(\{0\}) \geq 0$ and an arbitrary positive 
Radon measure $\mu$ satisfying eq.~\eqref{eq:ineq}. The constants
$a, b$ and the Radon measure $\mu$ are uniquely determined by
the function $f$.

Noting that 
$y/(y + r) = r\, [ 1/r - 1/(y + r) ]$,
we can also express the CBF $f$ in the following form
$$f(y) = b + a\, y + \int_0^\infty \left[1 - \e^{-z \, y}\right]
h(z) \, \dd z$$
where $h(z) :=  \int_{]0,\infty[}\; \e^{-r z}\, m(\dd r) \geq 0$
and $m(\dd r) := r\,\mu(\dd r)$ satisfies the inequality
$$ \int_{[0,\infty[}\; \frac{m(\dd r)}{r (r + 1)} < \infty$$
Let $\nu(\dd z) := h(z) \, \dd z$. 
We have  $$ \int_{[0,\infty[} \frac{z\,\nu(\dd z)}{1 + z} = \int_0^\infty \frac{z \, h(z)\, \dd z}{1 + z} = 
\int_{[0,\infty[} m(\dd r) \,[1/r - e^r \, \Gamma(0,r)]$$
Using the asymptotic properties of the incomplete
Gamma function \citep{Abramowitz} it is possible to prove
that the right-hand side is finite, hence the Radon measure 
$\nu(\dd z) := h(z) \, \dd z$ satisfies inequality~\eqref{eq:ineq2}.
We have thus proved an important theorem:
\begin{theorem} \label{thm:BFCBF}
Every CBF is a BF.
\end{theorem}
However $1 - \exp(-y)$ is a BF but not a CBF.

The simplest example of a CBF is $$ \varphi_a(y) :=
y/(y+a) \equiv y^2 \int_0^\infty \e^{-s y} \left[ 1 - \e^{-s a}\right]
\, \dd s $$
$a \geq 0$. It follows from eq.~\eqref{eq:CBFintegral} that 
every CBF $f$ which satisfies the conditions 
$f(0) = 0$ and $\lim_{y\rightarrow \infty} f(y)/y = 0$ 
is an integral superposition of the functions $\varphi_a$. The CBF
$\varphi_a$ corresponds to a Debye element defined by the relaxation
function $G_a(t) = \exp(-a \,t)$.

We shall need the following properties of CBFs:
\begin{theorem} \label{thm:CBF}\citep{Jacob01I,HanSerLANL}\\
\begin{enumerate}
\item $f$ is a CBF if and only if $y/f(y)$ is a CBF;
\item if $f, g$ are CBFs, then $f \circ g$ is a CBF.
\end{enumerate}
\end{theorem}
The second statement follows easily from Theorem~\ref{thm:analCBF}. 
\begin{remark} \label{rem:1}
$y^\alpha$ is a CBF if $0 < \alpha < 1$, 
because
\begin{multline*}
y^{\alpha-1} = \frac{1}{\Gamma(1-\alpha)} 
\int_0^\infty \e^{-y s}\, s^{-\alpha} \, 
\dd s =\\ \frac{1}{\Gamma(1-\alpha) \, \Gamma(\alpha)} 
\int_0^\infty \dd z \int_0^\infty  \,\e^{-y s} \, \e^{-z s} \dd s
\;z^{\alpha-1}= \\
\frac{\sin(\alpha \, \upi)}{\upi}
\int_0^\infty \frac{z^{\alpha-1}}{y + z}\, \dd z 
\end{multline*}
and thus 
$$y^\alpha = \frac{y\, \sin(\alpha \, \upi)}{\upi}
\int_0^\infty \frac{z^{\alpha-1}}{y + z}\, \dd z$$
\end{remark}

The sets of LICM functions and CBFs 
will be denoted by $\mathfrak{F}$ and $\mathfrak{C}$ 
respectively.

\section{Positivity of one- and three-dimensional solutions.}

Applying the results of the previous section, we get the following result:
\begin{theorem}
If $a \geq 0$ and the relaxation modulus $G$ is CM then the function $Q$
defined by eq.~\eqref{eq:defQ} is a CBF.\\
The mapping $(a, G) \in \overline{\mathbb{R}_+}\times \mathfrak{F} 
\rightarrow Q \in \mathfrak{C}$ defined by eq.~\eqref{eq:defQ} is bijective.
\end{theorem}

A one-dimensional solution of eq.\eqref{eq:Laplacetransformed} is given by
$$\tilde{u}_1(p,x) = U_1(p,\vert x\vert):=  
A(p) \, \exp(-B(p) \, \vert x \vert)$$ with
$B(p) = \rho^{1/2}\,p/Q(p)^{1/2}$ and $A(p) = g(p)/[2 B(p)]$
If $Q \in \mathfrak{C}$, then $Q(y)^{1/2}$ is a composition of two CBFs, namely 
$y^{1/2}$ (Remark~\ref{rem:1}) and $Q$, hence it is a CBF
by Theorem~\ref{thm:CBF}. The function $B(p)$ is a CBF by Theorem~\ref{thm:CBF}
and $1/B(p)$ is a CM function by Theorem~\ref{thm:BFCBF} and 
Corollary~\ref{corr:BF}. 

The amplitude of the solution of Problem~I is given by $A(p) = 1/[2\, p\, B(p)] + \dot{u}_0/[2 \,B(p)]$. The first term 
is a CM function because it is the product of two CM functions. The second 
term is also CM, hence $A(p)$ is CM. The amplitude of the solution of Problem~II 
$A(p) = 1/[2\, B(p)]$ is also CM.

For every fixed $x$ the function $\exp(-B(p)\, \vert x \vert)$ is the composition
of a CBF and the function $B$, which is a CBF and therefore a BF. 
By Theorem~\ref{thm:compBFCM} the function $\exp(-B(\cdot)\, \vert x \vert)$ is
CM. This proves that for $d = 1$ the solutions of Problem~II and Problem~I with
$u_0 = 0$ are non-negative.

In a three-dimensional space the solution $\tilde{u}_3$ of 
\eqref{eq:Laplacetransformed} is given by the equation 
$$\tilde{u}_3(p,x) = -\frac{1}{2 \upi r} 
\frac{\partial U_1(p,r)}{\partial r}$$
where $r = \vert x \vert$, so that 
\begin{equation}
\tilde{u}_3(p,x) = \frac{1}{4 \upi r} A(p) \, B(p) \, 
\exp(-B(p) \, \vert x \vert)
\end{equation}
But $A(p) \, B(p) = g(p)/2$. If $u_0 = 0$ then $g$ is CM. Hence
$\tilde{u}_3(\cdot,x)$ is the product of two CM functions and thus CM.

\section{Positivity of solutions in arbitrary dimension.}
\label{sec:multidim}

In an arbitrary dimension $d$ 
\begin{equation}
\tilde{u}_d(p,x) = \frac{g(p)}{(2 \upi)^d\, Q(p)} 
\int \e^{\ii \langle \vec{k},x\rangle}
\frac{1}{\rho\, p^2/Q(p) + \vert \vec{k} \vert^2} \dd_d k
\end{equation}
The above formula can be expressed in terms of MacDonald functions
by using eq.~(3) in Sec.~3.2.8 of \citet{Gelfand}:
\begin{equation}
\tilde{u}_d(p,x) = 
\frac{\rho^{d/4-1/2} \, g(p) \,p^{d/2-1}}{(2 \upi)^{d/2}\, Q(p)^{d/4+1/2}}
r^{-(d/2-1)}\,
K_{d/2-1}\left(B(p)\, r\right)
\end{equation}
where $B(p)$ is defined in the preceding section. 

The MacDonald function is given by the integral representation
\begin{equation}
K_\mu(z) = \int_0^\infty \exp(-z\, \cosh(s)) \, \cosh(\mu s) \, \dd s
\end{equation}
Since $\cosh(y)$ is a positive increasing function, it follows immediately 
that $K_\mu(z)$ is a CM function.

We shall need a stronger theorem on complete monotonicity of MacDonald functions. 
\begin{theorem} \citep{MillerSamko01}.\\
The function $z^{1/2}\, K_\mu(z)$ is CM for $\mu \geq 1/2$.
\footnote{The theorem is valid for $\mu \geq 0$, see \citet{MillerSamko01},
but we do not need this fact.}
\end{theorem}
The proof of this theorem requires a lemma.
\begin{lemma} \label{lem}
If $\alpha \geq 0$ then the function $\left( 1 + 1/x \right)^\alpha$ is CM.
\end{lemma}
\begin{proof}
We begin with $0 \leq \alpha < 1$. Setting $t = 1/(x y)$ we have that
\begin{multline*}
\frac{\alpha}{x^\alpha} \int_1^\infty \frac{\dd y}{y^{1+\alpha} \,(xy + 1)^{1-\alpha}} = 
\alpha \int_0^{1/x} \frac{t^{\alpha-1}}{(1/t + 1)^{1-\alpha}} \, \dd t = \\
\alpha \int_1^{1+1/x} u^{\alpha-1}\, \dd u = \left(1 + \frac{1}{x}\right)^\alpha
\end{multline*}
Since for each fixed value of $y > 0$ the function $(x y + 1)^{\alpha-1}$ is CM,
the function $\left(1 + 1/x\right)^\alpha$ ($x > 0)$ is also CM.

The function $1 + 1/x$ is CM, hence for every positive integer $n$ 
the function $(1 + 1/x)^n$ is CM. We can now decompose any positive non-integer 
$\alpha$ into the sum $\alpha = n + \beta$, where $n$ is a positive integer 
and $0 < \beta < 1$. Consequently
$$(1 + 1/x)^\alpha \equiv (1 + 1/x)^n \, (1 + 1/x)^\beta$$
is CM because it is a product of two CM functions.
\end{proof}

\begin{prooftheorem} 
For $\mu > -1/2$ the MacDonald function has the following integral representation:
\begin{equation}
z^{1/2} \, K_\mu(z) = \sqrt{\frac{\pi}{2}} \frac{1}{\Gamma(1/2-\mu)}
\e^{-z}\, \int_0^\infty \e^{-s} \, s^{\mu-1/2} \,
\left(1 + \frac{s}{2 z}   \right)^{\mu-1/2} \, \dd s, \qquad z > 0
\end{equation}
(\citet{GradshteinRhyzhik}, 8:432:8). By Lemma~\ref{lem}
the integrand of the integral on the right-hand side is CM if $\mu \geq 1/2$. 
Hence the integral is the limit of 
sums of CM functions, therefore itself a CM function. Consequently, the function 
$z^{1/2}\, K_\mu(z)$ is the product of two CM functions, and thus it is CM too.
\end{prooftheorem}

We now note that 
$\tilde{u}(p,x) = p^{(d-3)/2} \, g(p)\, F(p)$.  We shall prove that $F(p)$ is 
the product of two 
CM functions of the argument $p$, viz.  $Q(p)^{-(d+1)/4}$ and $L(z) := z^{1/2} \, K_{d/2-1}(z)$ 
with $z := B(p)\,r$,  as well as a positive factor independent of $p$.
\begin{lemma}
If $Q$ is a CBF and $\alpha > 0$, then $Q(p)^{-\alpha}$ is CM.
\end{lemma}
\begin{proof}
Let $n$ be the integer part of $\alpha$, $\alpha = n + \beta$,
$0 \geq \beta < 1$. $Q(p)^{-1}$ is CM (by Theorem~\ref{thm:BFCBF} and 
Corollary~\ref{corr:BF}) and therefore also $Q(p)^{-n}$ is CM. By 
Theorem~\ref{thm:CBF} the function $Q(p)^\beta$ is a CBF,
hence $1/Q(p)^\beta$ is CM. Consequently $Q(p)^{-\alpha}$ is CM.
\end{proof}
The lemma implies that the factor $Q(p)^{-(d+1)/4}$ is CM.
Since the function $L$ is CM and we have already proved that $B(p)$ is BF, 
Theorem~\ref{thm:compBFCM} implies that $L(B(p)\,r)$ is a CM function 
of $p$. For $d \leq 3$ the factor $p^{(d-3)/2}$ is also CM. Consequently,
for $d \leq 3$ the solution $u(t,x)$ of Problem~II is non-negative. The 
solution of the same problem with an arbitrary source of the form 
$s(t)\, \delta(x)$ 
and $s(t) \geq 0$ can be obtained by a convolution of two non-negative 
functions and therefore is also non-negative.

For $d \leq 5$  Problem~I with $u_0 = \dot{u}_0 = 0$  
 has a non-negative solution if $c_1 > 0$. 
For $d \leq 7$ Problem~I has a non-negative solution if
$c_1 = 0$ and $c_2 > 0$.

For $d > 3$ the fractional integral 
$$ \I^\alpha \, u(t,x) = \frac{1}{\Gamma(\alpha)} 
\int_0^t (t-s)^{\alpha-1} \, f(s) \, \dd s, \qquad t > 0$$
is non-negative provided $\alpha \geq (d-3)/2$ and $u_0 = 0$ or
provided $\alpha \geq (d - 1)/2$. 

We summarize these results in a theorem.

\begin{theorem}
In a viscoelastic medium of dimension $d \leq 3$ with a constitutive relation
$$\sigma = a\, \dot{e} + G(t)\ast\dot{e}, \qquad a \geq 0; \quad G \in 
\mathfrak{F}$$
Problem~II as well as Problem~I with the initial condition $u_0 = 0$ have 
non-negative solutions. 

Under the same assumptions but for an arbitrary dimension
$d > 3$ certain indefinite fractional time integrals of the solution are 
non-negative. For zero initial data Problem~I has a non-negative solution
if $d \leq 5$ and $c_1 > 0$, or if $d \leq 7$, $c_1 = 0$ and $c_2 > 0$.
\end{theorem}

\section{Positivity properties of vector-valued fields.}
\label{sec:vector}

It is interesting to examine the implications of CM relaxation kernels on positivity properties of vector fields.
We shall prove that in a simple model complete monotonicity of a relaxation kernel implies that the Green's function
is positive semi-definite.

 Unfortunately the tools developed in Sec.~\ref{sec:tools} 
fail for matrix-valued CM and Bernstein functions $Q(p)$ which do not commute with their derivatives.
In particular, the product of two non-commuting matrix-valued functions need not be a CM function and
the function $f\circ \mat{G}$, where $f$ is CM and $\mat{G}$ is a matrix-valued BF, need not be CM.

\begin{definition}
A matrix-valued function  $\mat{F}: \mathbb{R}_+ \rightarrow \mathbb{R}^{n\times n}$ is said to be a 
CM function if it is infinitely differentiable and the matrices $ (-1)^n\, \D^n \, \mat{F}(y)$ are 
positive semi-definite for all $y > 0$.
\end{definition}
\begin{definition}
A matrix-valued Radon measure $\mat{M}$ is said to be positive if the matrix 
$\langle \vec{v}, \int_{[0,\infty[}\; f(y) \, \mat{M}(\dd y)\, \vec{v} \rangle\geq 0$
for every vector $\vec{v} \in \mathbb{R}^n$ and every non-negative function $f$ 
on $\overline{\mathbb{R}_+}$ with compact support. 
\end{definition}
It is convenient to eliminate matrix-valued Radon measures by applying the following lemma \citep{HanSerVE}:
\begin{lemma} 
Every matrix-valued Radon measure $\mat{M}$ has the form $\mat{M}(\dd x)=\mat{K}(x)\, m(\dd x)$, where $m$ is a positive Radon measure, 
while $\mat{K}$ is a matrix-valued function defined, bounded and positive semi-definite on $\mathbb{R}_+$ except on a subset $E$
such that $m(E) = 0$. 
\end{lemma}

\begin{theorem} \citep{GripenbergLondenStaffans}
A matrix-valued function $\mat{F}: \mathbb{R}_+ \rightarrow \mathbb{R}^{n\times n}$ is CM if and only if it is
the Laplace transform of a positive matrix-valued Radon measure.
\end{theorem}
The following corollary will be applied to Green's functions:
\begin{corollary}
If $\tilde{\mat{R}}(p) := \int_0^\infty \e^{-p t} \, \mat{R}(t) \, \dd t$ is a matrix-valued CM function then
$\mat{R}(t)$ is positive semi-definite for $t > 0$.
\end{corollary}

\begin{definition}
A matrix-valued function  $\mat{G}: \mathbb{R}_+ \rightarrow \mathbb{R}^{n\times n}$ is said to be a Bernstein function (BF)
if $\mat{G}(y)$  is differentiable and positive semi-definite for all $y > 0$ and its derivative $\D \mat{G}$ is CM.
\end{definition}

\begin{definition}
A matrix-valued function  $\mat{H}: \mathbb{R}_+ \rightarrow \mathbb{R}^{n\times n}$ is said to be a complete Bernstein function (CBF)
if $\mat{H}(y) = y^2\, \tilde{\mat{G}}(y)$, where $\mat{G}$ is an $n \times n$ matrix-valued BF.
\end{definition}

The integral representation \eqref{eq:CBFintegral} of a CBF remains valid except that the Radon measure has to be replaced by a positive matrix-valued
Radon measure $\mat{N}(\dd r) = \mat{K}(r) \,\nu(\dd r)$:
\begin{equation} \label{eq:CBFintegralM}
\mat{H}(y) = \mat{B} +  y\, \mat{A} + y \int_{]0,\infty[}\; \frac{\mat{K}(r)\,\nu(\dd r)}{r + y}
\end{equation}
where the Radon measure $\nu$ satisfies the inequality 
\begin{equation} 
\int_{[0,\infty[}\; \frac{\nu(\dd r)}{1 + r} < \infty
\end{equation}
the matrix-valued function $\mat{K}(r)$ is positive semi-definite and bounded $\nu$-almost everywhere on $\mathbb{R}_+$
while $\mat{A}$, $\mat{B}$ are two positive semi-definite matrices. 
Every matrix-valued CBF $\mat{H}$ can be expressed in the form 
\begin{equation} \label{eq:CBFLICM-M}
\mat{H}(y) = y \, \tilde{\mat{F}}(y) +  y\, \mat{A}
\end{equation}
where $\mat{F}$ is a matrix-valued LICM function.

We now consider the following problem
\begin{equation}
\rho\, \D^2 \,\mathcal{G} = \mat{A}\,\D\, \mathcal{G} +  \mat{G}\ast \nabla^2\, \D\, \mathcal{G} + \delta(t) \,\delta(x)\,\mat{I}, 
\qquad t \geq 0, \quad x \in \mathbb{R}^d
\end{equation}
where $\mat{A}$ is a positive semi-definite $n \times n$ matrix and $\mat{G}$ is an $n \times n$ matrix-valued relaxation modulus.

If the relaxation modulus $\mat{G}$ is a CM matrix-valued function then the function $\mat{Q}(p) := p \, \tilde{\mat{G}}(p)$ 
is a matrix-valued CBF. The function $\mat{Q}$ is real and positive semi-definite, hence it is symmetric and has $n$ 
eigenvalues $q_n(p)$ and $n$ eigenvectors $\vec{e}_k$, $k = 1, \ldots, n$. We shall now assume that the eigenvectors are constant:
$$\mat{Q}(p) = \sum_{k=1}^n q_k(p) \, \vec{e}_k \otimes  \vec{e}_k$$

It is easy to see that the functions $q_k$, $k = 1,\ldots, n$, are CBFs.

The Laplace transform $\tilde{\mathcal{G}}(p,x)$ of the Green function is given by the
formula
\begin{multline*}
\tilde{\mathcal{G}}(p,x) = \frac{1}{(2 \upi)^d} \int \e^{\ii \langle \vec{k}, x \rangle} \left[ p^2\, \mat{I} + 
\vert \vec{k} \vert^2\, \mat{Q}(p)\right]^{-1} \, \dd_d k \equiv \\ \sum_{k=1}^n \frac{1}{(2 \upi)^d} \int \e^{\ii \langle \vec{k}, x \rangle} 
 \left[ p^2 + \vert \vec{k} \vert^2 \, q_k(p)\right]^{-1} \,\vec{e}_k \otimes  \vec{e}_k \equiv \\
\sum_{k=1}^n  g_k(p) \, \vec{e}_k \otimes  \vec{e}_k
\end{multline*}
where 
$$g_k(p) := \frac{\rho^{d/4-1/2} \,p^{d/2-1}}{(2 \upi)^{d/2}\, q_k(p)^{d/4+1/2}}\, 
r^{-(d/2-1)}\, K_{d/2-1}\left(B_k(p)\, r\right)$$
and
$B_k(p) := \rho^{1/2}\,p/q_k(p)^{1/2}$, $k = 1, \ldots, n$. Assume for definiteness that  $d \leq 3$. The argument of Sec.~\ref{sec:multidim} 
now leads to the conclusion that 
the functions $g_k$, $k = 1,\ldots, n$, are CM, hence the function $\tilde{\mathcal{G}}(\cdot,x)$ is a matrix-valued CM function and therefore
the Green function $\mathcal{G}(t,x)$ is positive semi-definite for $t \geq 0$, $x \in \mathbb{R}^d$. 
In particular, we have the following theorem:
\begin{theorem}
Let $\rho \in \mathbb{R}_+$, 
$d \leq 3$, $\vec{s}(t,x) = \delta(t)\, \delta(x)\, \vec{w}$, where $\vec{w} \in \mathbb{R}^n$. 

If $\mat{G}(s) = \sum_{k=1}^n G_k(s) \, \vec{e}_k \otimes  \vec{e}_k$ and $\mat{A} = \sum_{k=1}^n a_k \, \vec{e}_k \otimes  \vec{e}_k$
with CM functions $G_k$ and real numbers $a_k \geq 0$, $k = 1, \ldots, n$, then 
the solution $\vec{u}$ of the problem
$$\rho\, \D^2\, \vec{u} = \mat{A} \, \D\, \vec{u} + \mat{G}\ast \nabla^2 \, \D\,\vec{u} + \vec{s}(t,x), \qquad t\geq 0, \quad x \in \mathbb{R}^d$$ 
satisfies the inequality
\begin{equation}
\langle \vec{u}(t,x), \vec{w} \rangle \geq 0, \qquad t \geq 0, \quad x  \in \mathbb{R}^d.
\end{equation}
\end{theorem}

\section{Positivity in isotropic viscoelasticity.}

Consider now the Green's function $\mathcal{G}$ of a 3D isotropic viscoelastic medium. The function $\mathcal{G}$ 
is the solution of the initial-value 
problem:
\begin{multline}
\rho \, \D^2 \, \mathcal{G}_{kr}(t,x) = G(t)_{klmn}\ast \D\,\mathcal{G}_{mr,nl} + \delta(t)\, \delta(x) \, \delta_{kr}, \\ \qquad t > 0, 
\quad x \in \mathbb{R}^3, \quad k,r = 1,2,3
\end{multline}
with zero initial conditions, and
\begin{equation}
G_{klmn}(t) = \lambda(t) \, \delta_{kl}\, \delta_{mn} + \mu(t) \, \delta_{km}\,
\delta_{ln} + \mu(t) \, \delta_{kn}\, \delta_{lm}
\end{equation}
where the kernels $\lambda(t), \mu(t)$ are CM and $\rho \in \mathbb{R}_+$. The function $\tens{G}$ with the components $G_{klmn}$ 
takes values in the linear space $\Sigma$ of
symmetric operators on the space $S$ of symmetric $3 \times 3$ matrices. It is easy to see that under our hypotheses this function is CM:
$$ (-1)^n \, \langle \mat{e}_1, \tens{G}(t)\, \mat{e}_2 \rangle \geq 0 \qquad \text{for all $n = 0,1, 2 \ldots$} $$ 
for every $\mat{e}_1, \mat{e}_2 \in S$, where $\langle \mat{v}, \mat{w} \rangle := v_{kl}\, w_{kl}$ is the inner product on $S$. 

The Laplace transform $\tilde{\mathcal{G}}$ of $\mathcal{G}$ is given by the formula
$$\mathcal{G}(p,x) = \frac{1}{p\, \rho}  \left\{  
\nabla \otimes \nabla \Delta^{-1}\, F_\mathrm{L}(p,\vert x \vert)
+ \left[ \mat{I} -  \nabla \otimes \nabla \Delta^{-1}\right]\, F_\mathrm{T}(p,\vert x \vert) \right\}$$
where $\Delta := \nabla^2$, 
\begin{gather}
F_{\mathrm{L}}(p,r) := \frac{s_{\mathrm{L}}(p)^2 \,}{4 \upi r} \e^{-p^{1/2}\, s_{\mathrm{L}}(p)\, r} \label{eq:qq}\\
F_{\mathrm{T}}(p,r) := \frac{s_{\mathrm{T}}(p)^2 \,}{4 \upi r} \e^{-p^{1/2}\, s_{\mathrm{T}}(p)\, r} \label{eq:qr}
\end{gather}
and
\begin{gather}
s_{\mathrm{L}}(p)^2 := \frac{\rho}{\lambda(p) + 2 \mu(p)}\\
s_{\mathrm{T}}(p)^2 := \frac{\rho}{\mu(p)}
\end{gather}
Since $q_{\mathrm{L}}(p) = p/s_{\mathrm{L}}(p)$ and $q_{\mathrm{T}}(p) = p/s_{\mathrm{T}}(p)$ are CBFs, the functions $p/q_{\mathrm{L}}(p)^{1/2} =
p^{1/2}\, s_{\mathrm{L}}(p)$ and $p/q_{\mathrm{T}}(p)^{1/2} =
p^{1/2}\, s_{\mathrm{T}}(p)$ are BFs. hence the exponentials in eqs~(\ref{eq:qq}--\ref{eq:qr}) are CM functions of $p$.
 Moreover the functions $s_{\mathrm{L}}(p)^2$ and $s_{\mathrm{T}}(p)^2$ are CM. It follows that
the functions $F_{\mathrm{L}}(p,r)$ and $F_{\mathrm{T}}(p,r)$ are CM and therefore they are Laplace transforms of non-negative
functions $F_{\mathrm{L}}(t,r)$ and $F_{\mathrm{T}}(t,r)$. 
Their indefinite integrals $f_{\mathrm{L}}(t,r) := \int_0^t F_{\mathrm{L}}(s,r)\,
\dd s$ and $f_{\mathrm{T}}(t,r) := \int_0^t F_{\mathrm{L}}(s,r)\,
\dd s$ are also non-negative. 
The functions $h_{\mathrm{L}}(t,r) := \Delta^{-1}\, f_{\mathrm{L}}(t,r)\,
\dd s$, 
$h_{\mathrm{T}}(t,r) = \Delta^{-1}\, f_{\mathrm{T}}(t,r)$
 involve a convolution with a non-negative kernel and therefore 
are non-negative. The Green's function can be expressed in terms of these functions:
\begin{equation} \label{eq:GreenIsoVE}
\mathcal{G}(t,x) = \frac{1}{\rho} \left\{ \nabla \otimes \nabla h_{\mathrm{L}}(t,\vert x \vert) + 
\left[ \Delta\, \mat{I} - \nabla \otimes \nabla \right]\, h_\mathrm{T}(t,\vert x \vert) \right\}
\end{equation}

We shall use the notation $\vec{v} \geq 0$ if $v_k \geq 0$ for $k =1,2,3.$

\begin{theorem}
Let $\vec{u} = \nabla \phi + \nabla \times \vec{\psi}$ be the solution of the initial-value problem 
\begin{equation}
\rho\, \D^2 \,\vec{u} = \mat{G}\ast \nabla^2\, \D\, \vec{u} + \vec{s}(t,x), 
\qquad t \geq 0, \quad x \in \mathbb{R}^d
\end{equation}
with $\vec{u}(0,x) = 0 = \D \vec{u}(0,x)$ and
$\vec{s}(t,x) = \nabla f(t,x) + \nabla \times \vec{g}(t,x)$.

Then $\nabla f(t,x) \geq 0$ for all $t \geq 0$, $x \in \mathbb{R}^3$ implies 
that $\nabla \phi(t,x) \geq 0$ for all $t \geq 0$, $x \in \mathbb{R}^3$.

Similarly, $\nabla \times \vec{g}(t,x) \geq 0$ for all $t \geq 0$, $x \in \mathbb{R}^3$ implies 
that $\nabla \times \vec{\psi}(t,x) \geq 0$ for all $t \geq 0$, $x \in \mathbb{R}^3$.
\end{theorem}
\begin{proof}
Substitute $\vec{u} = \nabla \phi + \nabla \times \vec{\psi}$, 
$\vec{s} = \nabla f + \nabla \times \vec{g}$ in the formula
$$\vec{u}(t,x) = \int_0^t \int \mathcal{G}(t-s,x-y)\, \vec{s}(s,y) \dd_3 x \,
\dd s$$
where $\mathcal{G}$ is given by \eqref{eq:GreenIsoVE}. Noting that $\Delta^{-1}$
is a convolution operator commuting with $\nabla$ and 
$\nabla \,\Delta^{-1}\, \Div \vec{s} = \nabla f$ we have
$$\nabla \phi(t,x) = \frac{1}{\rho} \, 
\int f_{\mathrm{L}}(t-s,\vert x-y \vert)\, (\nabla f)(s,y)\, \dd_3 y$$
We now note that
$\left[\mat{I} - \Delta^{-1}\, \nabla \otimes \nabla\right] \, \vec{s} =
\vec{s} - \nabla f = \nabla \times \vec{g}$. Hence  
$$\nabla \times \vec{\psi}(t,x) = \frac{1}{\rho} \, \int f_{\mathrm{T}}(t-s,\vert x-y \vert)\, \nabla \times \vec{g}(s,y) \, \dd_3 y$$
The functions $f_{\mathrm{L}}$ and $f_{\mathrm{T}}$ are non-negative, hence the thesis follows.

\end{proof}

\section{Concluding remarks.}

A non-negative source term excites a non-negative viscoelastic pulse.
This result holds for scalar waves and for scalar potentials 
under the usual assumption that the stress response is determined 
by a CM relaxation modulus $G$ or by a Newtonian term
or both connected in parallel. The CM property of the relaxation modulus is
a fairly general property of real viscoelastic media, equivalent to
the assumption that the relaxation spectrum is non-negative. A generalization
of positivity for vector-valued viscoelastic fields in viscoelastic
media with the P class anisotropy \citep{HanCole} is sketched.

A particular example of a CBF is the rational function
$F(p) = R_N(p)/S_M(p)$, where $R_N$ and $S_M$ are two polynomials with 
simple negative roots $\lambda_k$, $k = 1, \ldots, N$, $\mu_l$,
$l = 1,\ldots, M$, $M = N$ or $N+1$ satisfying the intertwining conditions:
$$0 \leq \lambda_1 < \mu_1 < \ldots \mu_N \,[\, < \lambda_{N+1}\,]$$
(the last inequality is applicable only if $M = N + 1$)
\citep{Duff69}. A more general CBF is obtained by substituting in $F$
the
CBF $p^\alpha$, with $0 < \alpha < 1$:
$$ F_\alpha(p) = R_N\left(p^\alpha\right)/S_M\left(p^\alpha\right)$$
(Theorem~\ref{thm:CBF}). The choice of 
$Q = F_\alpha$ corresponds to a generalized Cole-Cole model
of relaxation. For $N = M = 1$ the original Cole-Cole model  
\citep{ColeCole,BagleyTorvik1} is recovered.

Anisotropic effects can be introduced by replacing 
the operator $\nabla^2$ by $g^{kl} \,\partial_k \, \partial_l$.
If $h_{kl}\, g^{lm} = \delta^k_m$ then
\begin{equation}
\tilde{u}_d(p,x) = \sqrt{\det{g}} \;
\frac{\rho^{d/4-1/2} \, g(p) \,p^{d/2-1}}{(2 \upi)^{d/2}\, Q(p)^{d/4+1/2}}
r^{-(d/2-1)}\,
K_{d/2-1}\left(\rho^{1/2}\, p   r/Q(p)^{1/2}\right)
\end{equation}
If $Q$ is a CBF then $u(t,x) \geq 0$. 
with $r := \left[h_{kl} \, x^k \, x^l\right]^{1/2}$,
cf \citet{Gelfand}.


\begin{thebibliography}{11}
\providecommand{\natexlab}[1]{#1}
\providecommand{\url}[1]{\texttt{#1}}
\providecommand{\urlprefix}{URL }
\expandafter\ifx\csname urlstyle\endcsname\relax
  \providecommand{\doi}[1]{doi:\discretionary{}{}{}#1}\else
  \providecommand{\doi}{doi:\discretionary{}{}{}\begingroup
  \urlstyle{rm}\Url}\fi
\providecommand{\eprint}[2][]{\url{#2}}

\bibitem[{Abramowitz \& Stegun(1970)}]{Abramowitz} 
Abramowitz, M. \& Stegun, I., 1970. 
\newblock \emph{Mathematical Tables}.
\newblock New York: Dover. 

\bibitem[{Bagley \& Torvik(1983)}]{BagleyTorvik1}
Bagley, R.~L. \& Torvik, P.~J., 1983.
\newblock A theoretical basis for the application of fractional calculus to
  viscoelasticity.
\newblock \emph{J. of Rheology} \textbf{27}, 201--210.

\bibitem[{Berg \& Forst(1975)}]{BergForst}
Berg, C. \& Forst, G., 1975.
\newblock \emph{Potential Theory on Locally Compact Abelian Groups}.
\newblock Berlin: Springer-Verlag.

\bibitem[{Cole \& Cole(1941)}]{ColeCole}
Cole, K.~S. \& Cole, R.~H., 1941.
\newblock Dispersion and absorption in dielectrics, {I}: {A}lternating current
  characteristics.
\newblock \emph{J. Chem. Phys.} \textbf{9}, 341--351.

\bibitem[{Duff(1969)}]{Duff69}
Duff, G.~F.~D., 1969.
\newblock Positive elementary solutions and completely monotonic functions.
\newblock \emph{J. Math. Anal. Appl.} \textbf{27}, 469--494.

\bibitem[{Gel'fand \& Shilov(1964)}]{Gelfand}
Gel'fand, I.~M. \& Shilov, G.~E., 1964.
\newblock \emph{Generalized Functions}, vol.~I.
\newblock New York: Academic Press.

\bibitem[{Gradshteyn \& Ryzhik(1994)}]{GradshteinRhyzhik}
Gradshteyn, I.~S. \& Ryzhik, I.~M., 1994.
\newblock \emph{Table of Integrals, Series and Products}.
\newblock New York: Academic Press.
\newblock $5^\mathrm{th}$ ed.

\bibitem[{Gripenberg \emph{et~al.}(1990)Gripenberg, Londen, \&
  Staffans}]{GripenbergLondenStaffans}
Gripenberg, G., Londen, S.~O., \& Staffans, O.~J., 1990.
\newblock \emph{Volterra Integral and Functional Equations}.
\newblock Cambridge: Cambridge University Press.

\bibitem[{Hanyga(2003)}]{HanCole}
Hanyga, A., 2003.
\newblock An anisotropic Cole-Cole viscoelastic model of 
seismic attenuation.
\newblock \emph{J. Comput. Acoustics}. \textbf{11}, 75--90.

\bibitem[{Hanyga \& Seredy\'{n}ska(2007)}]{HanSerVE}
Hanyga, A.\& Seredy\'{n}ska, M., 2007.
\newblock Relations between relaxation modulus and creep
compliance in anisotropic linear viscoelasticity.
\newblock \emph{J. Elasticity} \textbf{88}, 41--61
\newblock DOI 10.1007/s10659-007-9112-6.

\bibitem[{Jacob(2001)}]{Jacob01I}
Jacob, N., 2001.
\newblock \emph{Pseudo-{D}ifferential Operators and {M}arkov Processes},
  vol.~I.
\newblock London: Imperial College Press.

\bibitem[{Miller \& Samko(2001)}]{MillerSamko01}
Miller, K.~S. \& Samko, S.~G., 2001.
\newblock Completely monotonic functions.
\newblock \emph{Integr. Transf. and Spec. Fun.} \textbf{12}, 389--402.

\bibitem[{Seredy\'{n}ska \& Hanyga(2009)}]{HanSerLANL}
Seredy\'{n}ska, M. \& Hanyga, A., 2009.
\newblock Cones of material response functions in 1d and anisotropic linear
  viscoelasticity.
\newblock \emph{arXiv:0906.1983v1 [cond-mat]} .

\bibitem[{Widder(1946)}]{WidderLT}
Widder, D.~V., 1946.
\newblock \emph{The {L}aplace {T}ransform}.
\newblock Princeton: Princeton University Press.

\end{thebibliography}
\end{document}